\documentclass [twoside,reqno,12pt] {amsart}

\usepackage{amsfonts}
\usepackage{amssymb}
\usepackage{a4}

\newtheorem{thm}{Theorem}[section]

\newtheorem{prop}[thm]{Proposition}

\theoremstyle{definition}

\theoremstyle{remark}
\newtheorem{rem}{Remark}[section]

\numberwithin{equation}{section}

\renewcommand{\Im}{\hbox{Im}\,}

\newcommand{\C}{\mathbb{C}}

\newcommand{\N}{\mathbb{N}}

\newcommand{\R}{\mathbb{R}}

\def\tilde{\widetilde}
\def \bfo {\begin {eqnarray*} }
\def \efo {\end {eqnarray*} }
\def \ba {\begin {eqnarray*} }
\def \ea {\end {eqnarray*} }
\def \beq {\begin {eqnarray}}
\def \eeq {\end {eqnarray}}

\def \det {\hbox{det}}

\def \p {\partial}

\def\tilde{\widetilde}
\def \bfo {\begin {eqnarray*} }
\def \efo {\end {eqnarray*} }
\def \ba {\begin {eqnarray*} }
\def \ea {\end {eqnarray*} }
\def \beq {\begin {eqnarray}}
\def \eeq {\end {eqnarray}}

\def \det {\hbox{det}}

\def \p {\partial}

\begin{document}

 \title[Transmission eigenvalues for elliptic operators]{Transmission eigenvalues for elliptic operators}

\author[Hitrik]{Michael Hitrik}

\address
        {M. Hitrik,  Department of Mathematics\\
    UCLA\\
    Los Angeles\\
    CA 90095-1555\\
    USA }

\email{hitrik@math.ucla.edu}

\author[Krupchyk]{Katsiaryna Krupchyk}

\address
        {K. Krupchyk, Department of Mathematics and Statistics \\
         University of Helsinki\\
         P.O. Box 68 \\
         FI-00014   Helsinki\\
         Finland}

\email{katya.krupchyk@helsinki.fi}

\author[Ola]{Petri Ola}

\address
        {P. Ola, Department of Mathematics and Statistics \\
         University of Helsinki\\
         P.O. Box 68 \\
         FI-00014   Helsinki\\
         Finland}

\email{Petri.Ola@helsinki.fi}

\author[P\"aiv\"arinta]{Lassi P\"aiv\"arinta}

\address
        {L. P\"aiv\"arinta, Department of Mathematics and Statistics \\
         University of Helsinki\\
         P.O. Box 68 \\
         FI-00014   Helsinki\\
         Finland}

\email{Lassi.Paivarinta@rni.helsinki.fi}

\maketitle

\begin{abstract}

A reduction of  the transmission eigenvalue problem for multiplicative sign-definite perturbations of elliptic operators with constant coefficients to an eigenvalue problem for a non-selfadjoint compact operator is given. Sufficient conditions for the existence of transmission eigenvalues and completeness of generalized eigenstates for the transmission eigenvalue problem are derived.
In the trace class case, the generic existence of transmission eigenvalues  is established.

\end{abstract}

\section{Introduction}

Let $P_0(D)$ be an elliptic partial differential operator on $\R^n$, $n\ge 2$,  of order $m\ge 2$ with constant real coefficients,
\[
P_0(D)=\sum_{|\alpha|\le m} a_{\alpha}D^\alpha, \quad a_\alpha\in\R, \quad D_j=-i\frac{\partial}{\partial x_j},\quad j=1,\dots,n.
\]
Let $\Omega\subset \R^n$ be a bounded domain with a $C^\infty$-boundary and assume that 
$V\in C^\infty(\overline{\Omega}, \R)$ with $V>0$ in $\overline{\Omega}$.  The interior transmission problem associated to $P_0$ and $V$ is the following degenerate boundary value problem,
\begin{equation}
\label{eq_TE_acoustic}
\begin{aligned}
(P_0-\lambda)v=0 \quad &\text{in} \quad \Omega,\\
(P_0-\lambda(1+ V))w=0 \quad &\text{in} \quad \Omega,\\
 v-w \in H^{m}_0(\Omega).
\end{aligned}
\end{equation}
Here $H^m_0(\Omega)$ is the standard Sobolev space, defined as the closure of $C^\infty_0(\Omega)$ in the Sobolev space $H^m(\Omega)$.
We say that $\lambda\in \C$ is a transmission eigenvalue if the problem \eqref{eq_TE_acoustic} has non-trivial solutions $0\ne v\in L^2_{\textrm{loc}}(\Omega)$ and  $0\ne w\in L^2_{\textrm{loc}}(\Omega)$.

In the recent paper \cite{HitKruOlaPai}, we have studied the interior transmission problem and
transmission eigenvalues for multiplicative sign-definite perturbations of  linear partial differential operators with constant real coefficients.  Sufficient conditions for the discreteness of the set of transmission eigenvalues and for the existence of real transmission eigenvalues were obtained. In particular, in  the elliptic case, the set of transmission eigenvalues is discrete and in \cite{HitKruOlaPai}, the existence of real transmission eigenvalues was obtained for certain elliptic operators such as the biharmonic operator  and the Dirac system in $\R^3$.

The purpose of the present note is to point out an approach to the study of the transmission eigenvalues in the elliptic case, based on a reduction to 
the eigenvalue problem for a compact non-selfadjoint operator. By an application of Lidskii's theorem, we obtain sufficient conditions for 
the existence of (possibly complex) transmission eigenvalues, and the completeness of the set of the generalized eigenvectors, as well as 
demonstrate the generic existence of transmission eigenvalues. Let us mention explicitly that in this approach, we were directly inspired by the 
recent works \cite{AboRob, ChaHelLap04, HelRobWang, Rob2004},
where similar ideas in dealing with quadratic eigenvalue problems  have been used to study hypoelliptic partial differential operators
which are not analytic hypoelliptic.

The significance of transmission eigenvalues and of the interior transmission eigenvalue problem \eqref{eq_TE_acoustic}
comes from inverse scattering theory, and originally, this problem was introduced in \cite{ColMonk88}  in this context. The real transmission eigenvalues can be characterized as those values for which the scattering amplitude is not injective, see \cite{ColPaiSyl, HitKruOlaPai}.
Furthermore, in reconstruction algorithms of inverse scattering theory \cite{CakColbook, ColKir96, KirGribook}, transmission eigenvalues correspond to frequencies that one needs to avoid in the reconstruction procedure.

Recently there has been a large number of works devoted to the interior transmission eigenvalue problem \cite{CakColGint_complex, CakColHous10,  CakDroHou,  ColKirPai,   kir07, paisyl08}, with the major part being concerned with the case $P_0=-\Delta$.
The existing results establish the discreteness of the set of transmission eigenvalues, \cite{ColKirPai}, and give sufficient conditions for the 
existence of an infinite set of real transmission eigenvalues, \cite{CakDroHou, paisyl08}.  We would particularly like to mention the recent 
paper \cite{CakColGint_complex}, where the existence of complex transmission eigenvalues was shown, assuming that the perturbation $V$ in 
\eqref{eq_TE_acoustic} is constant and sufficiently small.

In this note, we have chosen to base our presentation on  the generalized acoustic wave equation $(P_0-\lambda(1+V))u=0$. Under the assumption that the full symbol of $P_0$ is non-negative,
all the results could equally well have been derived  for the following interior transmission problem associated to
the Schr\"odinger equation $(P_0+V-\lambda)u=0$,
\begin{align*}
(P_0-\lambda)v=0 \quad &\text{in} \quad \Omega,\\
(P_0+V-\lambda)w=0 \quad &\text{in} \quad \Omega,\\
 v-w \in H^{m}_0(\Omega).
\end{align*}

The structure of this note is as follows. In Section 2 we reduce the interior transmission problem to an eigenvalue problem for a compact non-selfadjoint operator in a suitable Schatten class. As a consequence of this reduction, in Section 3, we derive sufficient conditions for the existence of transmission eigenvalues and completeness of the generalized eigenstates. Finally, in Section 4, we show the generic existence of transmission eigenvalues in the trace class case.

\section{Reduction to an eigenvalue problem for a non-selfadjoint compact operator}

From \cite{HitKruOlaPai}, let us recall the following characterization of transmission eigenvalues.

\begin{prop}
\label{thm_equivalence}
Assume that $V\in C^\infty(\overline{\Omega}, \R)$ with $V>0$ in $\overline{\Omega}$. A complex number
$\lambda\ne 0$ is a transmission eigenvalue if and only if there exists $0\ne u\in H^{m}_0(\Omega)$ satisfying
\[
T_\lambda u:=(P_0-\lambda(1+V))\frac{1}{V}(P_0-\lambda)u=0\quad \text{in}\quad \mathcal{D}'(\Omega).
\]
\end{prop}

The question of deciding whether $0\ne \lambda\in \C$ is a transmission eigenvalue is therefore equivalent to finding a non-trivial solution $u\in H^m_0(\Omega)$ of the following quadratic eigenvalue problem
\begin{equation}
\label{eq_quadratic}
T_\lambda u=(A-\lambda B +\lambda^2 C)u=0,
\end{equation}
where
\[
A=P_0\frac{1}{V}P_0,\quad B=\frac{1}{V}P_0+P_0\frac{1}{V}+P_0,\quad C=1+\frac{1}{V}.
\]
Consider the following factorization
\begin{align*}
T_\lambda=C^{1/2}L_\lambda C^{1/2}, \quad L_\lambda&=\tilde A-\lambda \tilde B+\lambda^2,\\
\tilde A&=C^{-1/2}AC^{-1/2}, \quad \tilde B=C^{-1/2}BC^{-1/2}.
\end{align*}
In \cite{HitKruOlaPai} it was proved that the operator $\tilde A$, equipped
with the domain
\[
\mathcal{D}(\tilde A)=H^{2m}(\Omega)\cap H^m_0(\Omega),
\]
is a self-adjoint operator on $L^2(\Omega)$ with a discrete spectrum.  Here the regularity assumption on $V$ can be relaxed to 
$V\in C^N(\overline{\Omega})$, with $N$ being large enough but finite.

\begin{prop}
\label{prop_properties}

\begin{itemize}
\item[(i)]  The operator $\tilde A$ is positive, and $ \mathcal{D}(\tilde A^{1/2})=H_0^m(\Omega)$.
\item[(ii)]  The operators $\tilde B\tilde A^{-1/2}$ and $\tilde A^{-1/2}\tilde B$ are bounded in $L^2(\Omega)$.
\item[(iii)]  The operator $\tilde A^{-1/2}$ is in the  Schatten class $\mathcal{C}^p$ for $p>n/m$.
\end{itemize}

\end{prop}

We refer to \cite{Sim_book} for the definition and properties of the Schatten class operators.

\begin{proof}

(i). Let $u\in \mathcal{D}(\tilde A)\subset H^m_0(\Omega)$. Then
\[
(\tilde Au,u)=\int_{\Omega} \frac{1}{V}|P_0C^{-1/2}u|^2dx\ge C_{\Omega,V} \|u\|^2,\quad C_{\Omega,V}>0.
\]
Here the last inequality follows from the estimate \cite[Theorem 10.3.7]{horbookII}
\[
\|P_0(D)u\|\ge C_{\Omega}\|u\|, \quad u\in H^m_0(\Omega).
\]
We know from  \cite{HitKruOlaPai} that the form domain of the positive self-adjoint operator $\tilde A$ is $H_0^m(\Omega)$ and thus,
\[
\mathcal{D}(\tilde A^{1/2})=H_0^m(\Omega).
\]

(iii). The claim follows from the fact that  the inclusion map
\[
i:H^m_0(\Omega)\to L^2(\Omega)
\]
is in the Schatten class $\mathcal{C}^p$ for  $p>n/m$.  The latter can be concluded from the fact that  the operator  $(1-\Delta)^{-m/2}$ is  
in the Schatten class $\mathcal{C}^p$ for $p>n/m$,  on $L^2(\mathbb{T}^n)$, where $\mathbb{T}^n$ is the $n$-dimensional torus.
This concludes the proof of the proposition, as (ii) is clear.

\end{proof}

Notice that $0\ne \lambda\in \C$ is an eigenvalue of the quadratic eigenvalue problem $T_\lambda u=0$ with  an eigenstate $u\in H^m_0(\Omega)$  
if and only if $\lambda$ is an eigenvalue of the quadratic eigenvalue problem $L_\lambda v=0$ with $v=C^{1/2} u\in H^m_0(\Omega)$.

The holomorphic family $L_\lambda:\mathcal{D}(\tilde A)\to L^2(\Omega)$ is Fredholm of index $0$, invertible at $\lambda=0$. Thus, by the 
analytic Fredholm theory,
\[
L^{-1}_\lambda:L^2(\Omega)\to \mathcal{D}(\tilde A), \quad \lambda\in \C,
\]
is a meromorphic family of operators, with residues of finite rank.

Following \cite{Rob2004}, consider the closed operator
\[
\mathcal{A}=\begin{pmatrix} 0 & 1\\
-\tilde A & \tilde B
\end{pmatrix},
\]
acting
in the Hilbert space
\[
\mathcal{K}=\mathcal{D}(\tilde A^{1/2})\times L^2(\Omega)=H^m_0(\Omega)\times L^2(\Omega),
\]
equipped with the domain
\[
\mathcal{D}(\mathcal{A})=\mathcal{D}(\tilde A)\times \mathcal{D}(\tilde A^{1/2})=(H^{2m}(\Omega)\cap H^m_0(\Omega))\times H^m_0(\Omega).
\]
The spectrum of $\mathcal{A}$ is discrete, and as
\[
(\mathcal{A}-\lambda)^{-1}=\begin{pmatrix}
 L_\lambda^{-1}(\tilde B-\lambda) & - L_\lambda^{-1}\\
 L_\lambda^{-1}\tilde A & - L_\lambda^{-1}\lambda
\end{pmatrix},
\]
it follows that
$0\ne \lambda\in \C$ is a transmission eigenvalue if and only if  $\lambda$ is an eigenvalue of the operator $\mathcal{A}$.
The latter is equivalent  to the fact that $1/\lambda$ is an eigenvalue of the operator
\[
\mathcal{A}^{-1}=\begin{pmatrix}
\tilde A^{-1}\tilde B & -\tilde A^{-1}\\
1 & 0
\end{pmatrix}:\mathcal{K}\to \mathcal{K}.
\]
Given Proposition \ref{prop_properties}, it follows from \cite{Rob2004} that  $\mathcal{A}^{-1}$ is in the Schatten 
class $\mathcal{C}^p$ on $\mathcal{K}$, for $p>n/m$.

It will be more convenient to work in the Hilbert space $L^2(\Omega)\times L^2(\Omega)$ rather than  $\mathcal{K}$. To this end,  we introduce the operator
\[
T=\begin{pmatrix} \tilde A^{1/2} & 0\\
0 & 1
\end{pmatrix},
\]
which defines  an isomorphism
\[
T:\mathcal{K}\to L^2(\Omega)\times L^2(\Omega),
\]
and set
\begin{equation}
\label{eq_operator_P}
\mathcal{D}=T\mathcal{A}^{-1}T^{-1}=
\begin{pmatrix} \tilde A^{-1/2} \tilde B\tilde A^{-1/2} & -\tilde A^{-1/2}\\
\tilde A^{-1/2} &0
\end{pmatrix}: L^2(\Omega)\times L^2(\Omega)\to L^2(\Omega)\times L^2(\Omega).
\end{equation}
The operator $\mathcal{D}$ is in the Schatten class $\mathcal{C}^p$ on $L^2(\Omega)\times L^2(\Omega)$.
We summarize this section in the following result.

\begin{prop}
\label{prop_mathcal{D}}

A complex number $\lambda\ne 0$ is a transmission eigenvalue for {\rm (\ref{eq_TE_acoustic})} 
if and only if $1/\lambda$ is an eigenvalue of the operator $\mathcal{D}$ in {\rm (\ref{eq_operator_P})}. 
\end{prop}

\begin{rem}
It was shown in \cite{ColKirPai, ColPaiSyl,  HitKruOlaPai} that the set of transmission eigenvalues is discrete. The proof relied upon the analytic Fredholm theory. Our reduction of the transmission eigenvalue problem to the eigenvalue problem for the compact operator $\mathcal{D}$ gives another proof of the discreteness of  the set of transmission eigenvalues in the elliptic case.
\end{rem}

\section{Existence of transmission eigenvalues and completeness of transmission eigenstates}

In this section, we continue to work under the assumptions made in the beginning of the paper, namely that $P_0=P_0(D)$ is elliptic, $V\in C^\infty(\overline\Omega)$, $V>0$ on $\overline\Omega$, and $\p \Omega\in C^\infty$.

In the previous section, we have reduced the transmission eigenvalue problem to a spectral problem for the operator $\mathcal{D}\in \mathcal{C}^p$, $p>n/m$. Recall from  \cite{Sim_book} that this implies that $\mathcal{D}^p$ is of trace class, provided that $p\in \N$.

The following result is our main criterion for the existence of transmission eigenvalues. It is based on an application of   Lidskii's theorem, which we recall for the convenience of the reader, see e.g. \cite{GohGolKaa}:
let $\mathcal{A}$ be a trace class operator. Then
\[
\sum_j\mu_j(\mathcal{A})=\textrm{tr}(\mathcal{A}),
\]
where $\mu_j(\mathcal{A})$ are the non-vanishing eigenvalues of $\mathcal{A}$ counted with their algebraic multiplicities. In particular, if the spectrum $\textrm{spec}(\mathcal{A})=\{0\}$, then $\textrm{tr}(\mathcal{A})=0$.

\begin{prop}
\label{thm_trace_L}
 Assume that
 $p>n/m$, $p\in \N$, and $\emph{\textrm{tr}}(\mathcal{D}^{p})\ne 0$. Then the set of transmission eigenvalues is non-empty.

\end{prop}

\begin{proof}
Assume that spectrum  $\textrm{spec}(\mathcal{D})=\{0\}$. Then $\textrm{spec}(\mathcal{D}^p)=\{0\}$, since
\[
r(\mathcal{D}^p)=\lim_{n\to\infty}\|\mathcal{D}^{pn}\|^{1/n}=\lim_{n\to\infty}\|\mathcal{D}^{n}\|^{p/n}=r(\mathcal{D})^p=0,
\]
where $r(\mathcal{D})$ is the spectral radius of $\mathcal{D}$. By an application of  Lidskii's theorem, we get $\textrm{tr}(\mathcal{D}^{p})=0$, which contradicts the assumption of the proposition.

\end{proof}

\begin{rem}
In the case when $m>n$, the operator $\mathcal{D}$ is of trace class on $L^2(\Omega)\times L^2(\Omega)$, and $\textrm{tr}(\mathcal{D})=\textrm{tr}(\tilde A^{-1/2}\tilde B\tilde A^{-1/2})=\textrm{tr}(\tilde B\tilde A^{-1})$.

\end{rem}

\begin{rem}
In the case when $m>n/2$, the operator $\mathcal{D}$ is of Hilbert-Schmidt class
and
\[
\textrm{tr}(\mathcal{D}^{-2})=\textrm{tr}(\tilde A^{-1/2}(\tilde B\tilde A^{-1}\tilde B-2)\tilde A^{-1/2}).
\]
\end{rem}

The question of completeness of the eigenstates for the transmission eigenvalue problem for the Helmholtz equation has been posed in \cite{CakDroHou}. To the best of our knowledge, this issue remains unresolved in general.
We shall now give a sufficient condition for completeness.

Following~\cite{Rob2004} and~\cite{Markus}, we define the generalized eigenspace $\mathcal{E}_{\lambda_0}$ for  the transmission eigenvalue $\lambda_0\in \C$ as the closed linear space spanned by the vectors $(u_j)_{j=0}^\infty$, $u_j\in H^m_0(\Omega)$, where
\begin{align*}
&L_{\lambda_0}u_0=0, \quad u_0\ne 0,\\
& L_{\lambda_0}u_j+L'_{\lambda_0}u_{j-1}+\frac{1}{2}L''_{\lambda_0}u_{j-2}=0, \quad j=1,2,\dots.
\end{align*}
Here we set $u_{-1}=0$.

\begin{prop}
\label{thm_complete}
Assume that
 the set
 \begin{align*}
 \{\langle \tilde A^{-1/2} \tilde B \tilde A^{-1/2} u_0,u_0 \rangle_{L^2}-2i \emph{\Im} \langle \tilde A^{-1/2}v_0, u_0 \rangle_{L^2},\\
 u_0,v_0\in L^2(\Omega),\|(u_0,v_0)\|_{L^2\times L^2}=1
 \}
 \end{align*}
 lies in a closed angle with vertex at zero and opening $\pi/p$,  $p>n/m$. Then the space $\bigoplus_{\lambda\in \C}\mathcal{E}_\lambda$ is complete in $L^2(\Omega)$.

\end{prop}

\begin{proof}
It follows from \cite{Rob2004} and Proposition \ref{prop_mathcal{D}}
that
to show that the space $\bigoplus_{\lambda\in \C}\mathcal{E}_\lambda$  is complete in $L^2(\Omega)$, it suffices to verify that the space of generalized eigenvectors $\bigoplus_{\lambda}\mathcal{E}_\lambda[\mathcal{D}]$ of the operator $\mathcal{D}$ is complete in $L^2(\Omega)\times L^2(\Omega)$.
The latter can be obtained by an application  of  \cite[Theorem 3.1, Chapter X.3]{GohGolKaa}, which states that if the set
\[
\{\langle \mathcal{D}\varphi,\varphi \rangle_{L^2\times L^2}:\varphi\in L^2(\Omega)\times L^2(\Omega),\|\varphi\|_{L^2\times L^2}=1\}
\]
lies in a closed angle with vertex at zero and opening $\pi/p$, then the system of generalized eigenvectors of $\mathcal{D}$ is complete.
The claim follows.
\end{proof}

\begin{rem}
\label{rem_constant_potential}

In the case when $m>n$ and the operator 
$$
B  = \frac{1}{V} P_0 + P_0 \frac{1}{V} + P_0
$$
is non-negative on $H^m_0(\Omega)$,
it follows from Proposition \ref{thm_complete} that  the space $\bigoplus_{\lambda\in \C}\mathcal{E}_\lambda$ is complete in $L^2(\Omega)$.
In particular, if   $V=\textrm{const}>0$ in $\overline{\Omega}$ and  $P_0(\xi)\ge 0$, $\xi\in \R^n$,   an application of Proposition \ref{thm_complete} shows that
there exist infinitely  many transmission eigenvalues and the corresponding generalized transmission eigenstates form a complete system in $L^2(\Omega)$.
Notice that when $P_0=\Delta^2$ on $\R^3$,  the existence of infinitely many real transmission eigenvalues has been established in \cite{HitKruOlaPai}. The completeness of the generalized transmission eigenstates in the case of a constant potential for this operator seems to be a new observation.

\end{rem}

\begin{rem}

According to Proposition \ref{prop_mathcal{D}},
$0\ne \lambda\in \C$ is a transmission eigenvalue if and only if $1/\lambda$ is an eigenvalue of the operator $\mathcal{D}$.
Let us make explicit the connection between the generalized eigenvectors of $\mathcal{D}$ and the generalized transmission eigenstates. When doing so, since $\mathcal{D}=T\mathcal{A}^{-1}T^{-1}$, it will be convenient to consider the generalized eigenvectors of $\mathcal{A}$ directly.
Let
\[
\begin{pmatrix}
u_0\\
v_0
\end{pmatrix}\in H^m_0(\Omega)\times L^2(\Omega)
\]
be an eigenvector of $\mathcal{A}$ corresponding to $\lambda$, i.e.
\[
(\mathcal{A}-\lambda)\begin{pmatrix}
u_0\\
v_0
\end{pmatrix}=0\ \Longleftrightarrow \
  v_0=\lambda u_0,\ L_{\lambda}u_0=0,
\]
i.e. $u_0\in \mathcal{E}_{\lambda}$.
Let
\[
(\mathcal{A}-\lambda)^2\begin{pmatrix}
u_1\\
v_1
\end{pmatrix}=0.
\]
This is equivalent to the fact that
\[
(\mathcal{A}-\lambda)\begin{pmatrix}
u_1\\
v_1
\end{pmatrix}=\begin{pmatrix}
u_0\\
v_0
\end{pmatrix}
\]
is an eigenvector of $\mathcal{A}$. The latter is equivalent to the fact that
\[
v_1=u_0+\lambda u_1,\quad L_{\lambda}u_1+L'_{\lambda}u_0=0,
\]
i.e. $u_1\in \mathcal{E}_{\lambda}$.
Continuing in the same fashion, for $j=2,3,\dots$, we have
\[
(\mathcal{A}-\lambda)^{j+1}\begin{pmatrix}
u_j\\
v_j
\end{pmatrix}=0
\]
is equivalent to
\[
v_j=u_{j-1}+\lambda u_j,\quad L_{\lambda_0}u_j+L'_{\lambda }u_{j-1}+u_{j-2}=0,
\]
i.e. $u_j\in \mathcal{E}_{\lambda}$. This shows that the first components of the generalized eigenvectors of $\mathcal{A}$, corresponding to the eigenvalue $\lambda$, are given by the generalized transmission eigenstates, corresponding  to the transmission eigenvalue $\lambda$, and vice versa.

\end{rem}

\section{Generic existence of transmission eigenvalues in the trace class case}

In this section, we let $P_0=P_0(D)$ be a formally selfadjoint elliptic operator with constant coefficients of order $m$, with 
$m>n$, $V\in C^N(\overline{\Omega})$ where $N$ is large enough fixed, and $\p \Omega\in C^\infty$.
Let us introduce the following open connected subset of the real Banach space $C^N(\overline{\Omega}, \R)$,
\[
\mathcal{E}=\{V\in C^N(\overline{\Omega},\R):V>0\}.
\]
When $V\in \mathcal{E}$, we shall be concerned with the quantity $\textrm{tr}(\mathcal{D})=\textrm{tr}(\tilde B\tilde A^{-1})$.
In order to indicate the dependence of the operators $\tilde A$ and $\tilde B$ on the potential, we shall write
\begin{align*}
q=\frac{1}{V},&\quad V\in \mathcal{E},\quad  A_q=P_0qP_0,\quad  B_q=qP_0+P_0q+P_0,\\
\tilde A=\tilde A_q&=(1+q)^{-1/2}A_q(1+q)^{-1/2},\quad \tilde B= \tilde B_q=(1+q)^{-1/2}B_q(1+q)^{-1/2}.
\end{align*}
Using the cyclicity property of the trace, we have
\[
\textrm{tr}(\tilde B_q \tilde A_q^{-1})=\textrm{tr}((1+q)^{-1/2}B_q A_q^{-1}(1+q)^{1/2})=\textrm{tr}(B_q A_q^{-1}).
\]

\begin{thm}  Assume that $m>n$ and that $P_0(\xi)\ge 0$, $\xi\in \R^n$.  Then the set
\[
\mathcal{F}=\{V\in \mathcal{E}:\emph{tr}(B_qA^{-1}_q)\ne 0\}
\]
is open and dense in $\mathcal{E}$.
\end{thm}

The theorem above and Proposition \ref{thm_trace_L} imply the existence of transmission eigenvalues in the trace class case, for an open and dense
set of potentials.

\begin{proof}
Let us first show that the set $\mathcal{F}$ is open. To this end it suffices to prove that the function
$
 V\mapsto \textrm{tr}(B_qA^{-1}_q)$
is continuous on $\mathcal{E}$ in the topology of $C^N(\overline{\Omega}, \R)$. We shall show that the map
$V\mapsto B_qA^{-1}_q$  is continuous, with values in the space of trace class operators.

Let $V_j\to V$ in $\mathcal{E}$. Then  $\p^\alpha q_j\to \p^\alpha q$ uniformly on $\overline{\Omega}$ for  $|\alpha|\le N$.  Let us write
\begin{equation}
\label{eq_op_1}
\begin{aligned}
 B_{q_j}A^{-1}_{q_j}-B_qA^{-1}_q=(B_{q_j}-B_q)A^{-1}_{q_j}+B_q(A_{q_j}^{-1}-A^{-1}_q)
\end{aligned}
\end{equation}
When treating the first term in the right hand side of \eqref{eq_op_1}, we have 
\[
(B_{q_j}-B_q)A^{-1}_{q_j}=((q_j-q)P_0+P_0(q_j-q))A^{-1}_{q_j}.
\]
Thus,
\begin{align*}
&\|(q_j-q)P_0A^{-1}_{q_j}\|_{\textrm{tr}}\le \|q_j-q\|_{L^\infty}\|P_0A^{-1/2}_{q_j}\|\|A^{-1/2}_{q_j}\|_{\textrm{tr}},\\
&\|P_0(q_j-q)A^{-1}_{q_j}\|_{\textrm{tr}}\le \|P_0\|_{H^m_0\to L^2} \|q_j-q\|_{H^m_0\to H^m_0}\|A^{-1/2}_{q_j}\|_{L^2\to H^m_0}
\|A^{-1/2}_{q_j}\|_{\textrm{tr}},
\end{align*}
and hence, both expressions tend to zero as $j\to \infty$, provided that $N\ge m$.
When considering the second term in the right hand side of \eqref{eq_op_1}, we write, using the resolvent identity,
\begin{align*}
B_q(A_{q_j}^{-1}-A^{-1}_q)&=B_qA^{-1}_{q_j}(A_q-A_{q_j})A^{-1}_q\\
&=(B_qA^{-1/2}_{q_j})(A^{-1/2}_{q_j}P_0)(q-q_j)(P_0A^{-1/2}_q)A^{-1/2}_q.
\end{align*}
The trace class norm of the above expression is easily seen to  vanish as $j\to\infty$. If follows that the set $\mathcal{F}$ is open.

Let us now show that the set  $\mathcal{F}$ is dense in $\mathcal{E}$.
Let $V_0\in \mathcal{E}$ be fixed. Then there exists a complex neighborhood $U\subset C^N(\overline{\Omega}, \C)$ of $V_0$ such that the map
\begin{equation}
\label{eq_complex_pot}
U\to \C,\quad
V\mapsto\textrm{tr}(B_qA^{-1}_q)
\end{equation}
is well-defined on $U$. This follows from the fact that the operator
\[
A_q:H^{2m}(\Omega)\cap H^m_0(\Omega)\to L^2(\Omega), \quad q=\frac{1}{V},
\]
is bijective for $V\in U$, since the operator norm of
\[
P_0\Im q P_0 A_{\mathrm{Re}\, q}^{-1}:L^2(\Omega)\to L^2(\Omega)
\]
is small.

We claim that the map \eqref{eq_complex_pot} is analytic. Since the arguments above show that the map \eqref{eq_complex_pot} is continuous, 
it therefore suffices to check the weak analyticity, \cite{postru_book}. To this end let $q_1=1/V_1$, $V_1\in U$, and $q_2$ be arbitrary, and consider the function
\begin{equation}
\label{eq_holom}
 z\mapsto \textrm{tr}(B_{q_1+zq_2}A_{q_1+zq_2}^{-1})
\end{equation}
for $z$ near $0\in \C$. We have the convergent power series expansion
\[
A^{-1}_{q_1+zq_2}=A^{-1}_{q_1}\sum_{k=0}^\infty(-z)^k(P_0q_2P_0A_{q_1}^{-1})^k
\]
for $z$ near $0\in \C$. Since the operator $(B_{q_1}+zB_{q_2})A^{-1}_{q_1}$ is of trace class, the operator
\[
B_{q_1+zq_2}A_{q_1+zq_2}^{-1}
\]
is given by a power series in $z$ which converges in the trace class norm.
Thus, it follows that the map \eqref{eq_holom} is holomorphic near $0\in \C$.

We therefore conclude that the map
\[
V\mapsto\textrm{tr}(B_qA^{-1}_q)
\]
is real-analytic on $\mathcal{E}$.
We furthermore know from Remark \ref{rem_constant_potential} that it does not vanish identically, for it is positive at $V=1$.
Since $\mathcal{E}$ is connected, given $V_0\in \mathcal{E}$ it follows that for any neighborhood of $V_0$ there are points $V$ for which $\textrm{tr}(B_qA^{-1}_q)\ne 0$.
This completes the proof.

\end{proof}

Finally, concerning counting estimates for transmission eigenvalues, we have the following simple result.
\begin{prop}
Let $m>n$. Then the number of transmission eigenvalues in the disk of radius $R$ is $\mathcal{O}(R^2)$.
\end{prop}

\begin{proof}
Recall that $0\ne\lambda\in \C$ is a transmission eigenvalue if and only if $1/\lambda$ is an eigenvalue of the operator $\mathcal{D}$ given by \eqref{eq_operator_P}.
The latter is equivalent to the fact that the operator
\[
I-\lambda(\tilde A^{-1/2}\tilde B\tilde A^{-1/2}) +\lambda^2 \tilde A^{-1}:L^2(\Omega)\to L^2(\Omega)
\]
is not invertible. Here $\tilde A^{-1/2}\tilde B\tilde A^{-1/2}$ and $\tilde A^{-1}$ are of trace class.
Thus, the latter is equivalent to the fact that
\[
\det(I-\lambda(\tilde A^{-1/2}\tilde B\tilde A^{-1/2}) +\lambda^2 \tilde A^{-1})=0.
\]
The function
\[
f(\lambda)=\det(I-\lambda(\tilde A^{-1/2}\tilde B\tilde A^{-1/2}) +\lambda^2 \tilde A^{-1})
\]
is entire holomorphic. Therefore, the number $N(R/2)$ of its zeros in the disk of radius $R/2$ can be estimated by
Jensen's formula,
\[
N(R/2)\le \frac{1}{\log 2}(\max_{|\lambda|=R }\log|f(\lambda)|-\log|f(0)|)=\mathcal{O}(R^2).
\]
Here we have used that  $|f(\lambda)|\le e^{C|\lambda|^2}$ with some constant $C$ and $|\lambda|\ge 1$.

\end{proof}

\section{Acknowledgements}
 The research of M.H. was partially supported by the NSF grant DMS-0653275 and he is grateful to the Department of Mathematics and Statistics at the University of Helsinki for the hospitality. The research of K.K. was financially supported by the
Academy of Finland (project 125599).
The research of P.O. and L.P. was financially supported by Academy of Finland Center of Excellence programme 213476.


\begin{thebibliography} {1}

\bibitem{AboRob} Aboud, F.,  Robert, D., \emph{Asymptotic expansion for nonlinear eigenvalue problems},  J. Math. Pures Appl. (9)  \textbf{93}  (2010),  no. 2, 149--162.

\bibitem{CakColbook}
Cakoni, F.,  Colton, D., \emph{Qualitative Methods in Inverse Scattering Theory}, Springer,
Berlin, 2006.

\bibitem{CakColGint_complex}
Cakoni, F.,  Colton, D., and Gintides, D., \emph{The interior transmission eigenvalue problem}, preprint, 2010.

\bibitem{CakColHous10}
Cakoni, F.,  Colton, D., and Haddar, H.,  \emph{The interior transmission problem  for regions with cavities}, SIAM J. Math. Analysis \textbf{42} (2010), no 1, 145--162.

\bibitem{CakDroHou} Cakoni, F.,  Drossos,  G., and Houssem,  H., \emph{The existence of an infnite discrete set of transmission eigenvalues}, SIAM J. Math. Analysis, \textbf{42} (2010), no 1, 237--255.

\bibitem{ChaHelLap04}
Chanillo, S., Helffer, B., and Laptev, A., \emph{Nonlinear eigenvalues and analytic hypoellipticity}, J. Funct. Anal. \textbf{209} (2004), no. 2, 425--443.

\bibitem{ColKir96}
Colton, D., Kirsch,  A., \emph{A simple method for solving inverse scattering problems
in the resonance region}, Inverse Problems \textbf{12} (1996), 383--393.

\bibitem{ColKirPai}
Colton, D., Kirsch, A. and P\"aiv\"arinta, L.,  \emph{Far-field patterns for acoustic waves in an inhomogeneous medium}, SIAM J. Math. Anal. \textbf{20} (1989), no. 6, 1472--1483.

\bibitem{ColMonk88}
Colton, D., Monk P., \emph{The inverse scattering problem for acoustic waves in an inhomogeneous medium}, Quart. Jour. Mech. Applied Math, \textbf{41} (1988), 97--125.

\bibitem{ColPaiSyl}
Colton, D.,  P\"aiv\"arinta L. and Sylvester, J., \emph{The Interior Transmission Problem}, Inverse Problems and Imaging, Vol. 1 (\textbf{1}) (2007), 13--28.

\bibitem{GohGolKaa}
Gohberg, I.; Goldberg, S.; Kaashoek, M. A., \emph{Classes of linear operators}. Vol. I. Operator Theory: Advances and Applications, 49. Birkh\"auser Verlag, Basel, 1990, 468 pp.

\bibitem{HelRobWang}
Helffer, B., Robert, D., and Wang, X. P., \emph{Semiclassical analysis of a nonlinear eigenvalue problem and nonanalytic hypoellipticity}, Algebra i Analiz \textbf{16} (2004), no. 1, 320--334; translation in
St. Petersburg Math. J. \textbf{16} (2005), no. 1, 285--296.

\bibitem{HitKruOlaPai}
Hitrik M., Krupchyk K., Ola, P., and P\"aiv\"arinta, \emph{Transmission eigenvalues  for operators with constant coefficients}, Preprint, http://arxiv.org/abs/1004.5105.

\bibitem{horbookII}
H\"ormander, L., \emph{The analysis of linear partial differential operators. II. Differential operators with constant coefficients}. Classics in Mathematics. Springer-Verlag, Berlin, 2005, 392 pp.


\bibitem{kir07}
Kirsch, A.,  \emph{An integral equation approach and the interior transmission problem for Maxwell's equations},  Inverse Probl. Imaging  \textbf{1}  (2007),  no. 1, 159--179.

\bibitem{KirGribook}
Kirsch, A.,   Grinberg, N., \emph{The Factorization Method for Inverse Problems}, Oxford University Press, Oxford, 2008.

\bibitem{Markus}
Markus, A., \emph{Introduction to the spectral theory of polynomial operator pencils}, American Mathematical Society, Providence, RI, 1988. 

\bibitem{paisyl08}
P\"aiv\"arinta, L., Sylvester, J., \emph{Transmission eigenvalues}, SIAM J. Math. Anal.,  \textbf{40} (2008), no. 2, 738--753.

\bibitem{postru_book}
P\"oschel, J.,  Trubowitz, E., \emph{Inverse spectral theory},
Pure and Applied Mathematics, 130. Academic Press, Inc., Boston, MA, 1987, 192 pp.

\bibitem{Rob2004}
Robert, D.,  \emph{Non-linear eigenvalue problems},  Mat. Contemp.  \textbf{26}  (2004), 109--127.

\bibitem{Sim_book}
Simon, B., \emph{Trace ideals and their applications}, London Mathematical Society, Lecture Note Series,
Vol. 35, Cambridge University Press, Cambridge, 1979.


\end{thebibliography}
\end{document}